%% file: FinalMSS.tex
\documentclass[10pt,letter,final]{IEEEtran}
\usepackage{amsthm,amssymb,graphicx,multirow,amsmath,color,amsfonts}
\usepackage[update,prepend]{epstopdf}%
\usepackage[noadjust]{cite}
\usepackage{tikz}
\usetikzlibrary{arrows,calc}		
\usepackage{bbm} 
\usepackage{pdfpages}
\usepackage{tabulary}
\usepackage{multirow}
\usepackage{comment}
\usepackage{dsfont}
\usepackage{pgf,tikz}
\usepackage{bigints}
\usepackage{mathtools}
\usepackage[english]{babel}
\usepackage[autostyle]{csquotes}
\usepackage{textcomp}
\usepackage{varwidth}

\newcommand{\cdf}{\mathtt{CDF}}

\newcommand\numberthis{\addtocounter{equation}{1}\tag{\theequation}}

\makeatletter
\renewcommand{\env@cases}[1][@{}l@{\quad}l@{}]{%
  \let\@ifnextchar\new@ifnextchar
  \left\lbrace
  \def\arraystretch{1.2}%
  \array{#1}%
}
\makeatother

\include{notation}

\allowdisplaybreaks 

\begin{document}
\graphicspath{{./Figures/}}
\title{\huge Analysis of Transmission Rate of Wireless Networks under the Broadcast Approach with Continuum of Transmission Layers}
\author{
Praful D. Mankar and Harpreet S. Dhillon 
\thanks{The authors are with Wireless@VT, Bradley Department of Electrical and Computer Engineering, Virginia Tech, Blacksburg, VA. Email: \{prafuldm, hdhillon\}@vt.edu. \hfill 
}
}

\IEEEaftertitletext{\vspace{-2\baselineskip}}

\maketitle
\begin{abstract} In this letter, we characterize the performance of broadcast approach with continuum of transmission layers in random wireless networks where the channel state information (CSI) is assumed to be known only at the receiver. By modeling the transmitter-receiver pairs using bipolar Poisson point process, we derive analytical expressions for the mean and variance of achievable transmission rate and the network transmission capacity under broadcast approach. Our analysis shows that the broadcast approach provides better mean transmission rate with lower variance compared to outage strategy.
\end{abstract}
\vspace{-.15cm}
\vspace{-.35cm}
\section{Introduction}
\label{sec:Introduction}
In the absence of CSI at the transmitter, {\em outage}~\cite{BiglieriProakisShamai} and {\em broadcast}~\cite{Cover_BroadcastChannel,Shamai1997} are two well-known transmission strategies. In the {\em outage strategy}, the transmitter sends data frames at a fixed rate and the receiver is able to decode them successfully when the channel state is above a certain threshold. Otherwise, the transmission fails, which is defined as {\em outage}. Please refer to \cite{BiglieriProakisShamai} for the achievable capacity under {\em outage strategy} for a point-to-point link. On the other hand, \textit{broadcast strategy} adapts its transmission rate with the channel states without the need for CSI at the transmitter \cite{Cover_BroadcastChannel,Shamai1997}. In particular, the transmitter sends the information encoded in multiple layers and the receiver tries to decode the maximum number of layers depending upon the current channel state. The achievable rate using broadcast approach with continuum of transmission layers for a point-to-point link is studied in \cite{Shamai1997,Shamai2003}. 

The outage strategy has been widely used to analyze key performance metrics, such as transmission rate, packet delay, and network capacity, in large-scale wireless networks using tools from stochastic geometry, e.g., see \cite{Baccelli_Aloha2006,WeberAndrews2010,AndBacJ2011,Haenggi2013}. Some newer perspectives on enhancing spectral efficiency in {\em single layer} transmission strategies as well as some newer definitions of key metrics are explored recently in \cite{alammouri2019unified} and \cite{DiRenzo2018}, respectively.  However, it is somewhat surprising to note that the similar analysis for the broadcast strategy in the large-scale network setting has not been performed yet, which is the main focus of this letter. 

{\em Contributions:} In this letter, we extend the analysis of \cite{Shamai1997,Shamai2003} to a large-network setting and present the transmission rate analysis for wireless networks under broadcast strategy. In particular, by modeling the locations of transmitter-receiver pairs using bipolar Poisson point process (PPP), we analyze the mean and variance of transmission rate of a typical link and the transmission capacity (TC) of  network. We define TC as the product of the mean transmission rate and the maximum spatial density of transmitter under rate outage constraint. Through numerical comparisons, we observe that the broadcast approach performs better in terms of  mean and variance of transmission rate as compared to the outage strategy.
\section{System Model}
We consider a large ad-hoc network modeled as a bipolar PPP, wherein the transmitters are distributed according to a homogeneous PPP $\Phi$ with density $\lambda$ and each transmitter has a designated receiver at a distance $R_0$ in a random direction that is independent and identically distributed (i.i.d.) for all links.
We assume saturated queues, i.e., each transmitter always has a packet in its queue for transmission. 
The analysis will be performed for the {\em typical link} of this bipolar PPP whose transmitter is placed at $\mathbf{x}_o\equiv[0,-R_0]$ and receiver at the origin $o$. Since, by Slivnyak's theorem, the reduced Palm distribution of a PPP is the same as its original distribution, the locations of the other transmitters in the network will simply form the PPP $\Phi$. 
Thus, considering the single-input-single-output (SISO) channel, we have  $z_i = $
\begin{equation}
\sqrt{P \ell(\mathbf{x}_0)}h_{\mathbf{x}_0,i}y_{\mathbf{x}_0,i} + \sum\nolimits_{\mathbf{x}_k\in\Phi}{\sqrt{P \ell(\mathbf{x}_k)}h_{\mathbf{x}_k,i}y_{\mathbf{x}_k,i}} + n_i
\label{eq:Channel}
\end{equation}
where $\{z_i\}$ are the received symbols at the receiver of the typical link, $\{y_{\mathbf{x}_k,i}\}$ are the transmitted symbols from transmitter $\mathbf{x}_k$, $\{h_{\mathbf{x}_k,i}\}$ are fading coefficients (assumed to be i.i.d. across all links), and $\{n_i\}\sim\mathcal{CN}(0,\sigma_N^2)$ are additive noise samples. Assuming Rayleigh fading, $|h_{\mathbf{x}_k,i}|^2 \sim \exp(1)$. 
Further, assuming the power-law path-loss, we have $\ell(\mathbf{x}_k)=\|\mathbf{x}_k\|^{-\alpha}$ where $\alpha>2$ is the path-loss exponent. 
For this setup, we will characterize the performance of broadcast strategy and compare it with that of the outage strategy assuming that the channel state information is available only to the receiver. We define the {\em channel state} $S$ as
\begin{equation}
S=\frac{|h_{\mathbf{x}_0,i}|^2R_0^{-\alpha}}{\sum_{\mathbf{x}_k\in\Phi}|h_{\mathbf{x}_k,i}|^2\|\mathbf{x}_k\|^{-\alpha}P+\sigma_N^2},
\label{eq:ChannelState}
\end{equation}
which will be useful in constructing transmission layers for the broadcast approach. We first characterize the cumulative distribution function ($\cdf$) of $S$ conditioned on the locations of interfering transmitters, $\Phi$, in the following Lemma. 
\begin{lemma}
\label{lemma:ConditionalCDF_S}
$\cdf$ of channel state $S$ conditioned on $\Phi$ is 
\begin{equation}
\mathbb{P}[S>s|\Phi] = \prod\nolimits_{\mathbf{x}_k\in\Phi}\frac{\exp(-sR_0^\alpha {\sigma}_N^2)}{1+sPR_0^\alpha\|\mathbf{x}_k\|^{-\alpha}}.
\label{eq:CDF_S_condPhi}
\end{equation} 
\begin{proof}
We obtain \eqref{eq:CDF_S_condPhi} using the fact that $|h_{\mathbf{x}_k,i}|^2 \sim \exp(1)$ and $\mathbb{P}[|h_{\mathbf{x}_o,i}|^2>sR_0^\alpha(\sigma_N^2+\sum_{\mathbf{x}_k\in\Phi}|h_{\mathbf{x}_{k},i}|^2\|\mathbf{x}_{k}\|^{-\alpha})\mid\Phi]$ $=\prod_{\mathbf{x}_k\in\Phi}\mathbb{E}\left[\exp(-sR_0^\alpha(\sigma_N^2+|h_{\mathbf{x}_{k},i}|^2P\|\mathbf{x}_{k}\|^{-\alpha}))\mid\Phi\right]$.
\end{proof}
\end{lemma}

{\em Outage strategy.} Since channel state $S$ is assumed to be unknown at the transmitter, a reasonable baseline for comparisons with the broadcast strategy is the well-known outage strategy in which the transmitter sends data blocks at fixed rate $\ln(1+\beta)$ and receiver is able to decode the blocks when the received signal-to-interference-plus-noise ratio, $\sinr_{os}=SP$, is above threshold $\beta$. Another reason for considering this as the baseline is the fact that its performance is well-understood in the aforementioned setting. For instance, the mean transmission rate for outage strategy is 
\begin{equation}
R_{os}(\lambda,\beta)=\mathbb{P}[\sinr_{os}\geq\beta]\ln(1+\beta),~~\text{where}
\label{eq:Rate_OS}
\end{equation}  
$\mathbb{P}[\sinr_{os}\geq\beta]=\mathbb{E}_{\Phi}[\mathbb{P}[SP>\beta|\Phi]]=\exp(-\pi\lambda R_0^2 \beta^\delta \mathcal{Z}_\delta - R_0^\alpha \tilde{\sigma}_N^2 \beta)$, $\delta=\frac{2}{\alpha}$, $\tilde{\sigma}_N^2={\sigma}_N^2/P$ and $\mathcal{Z}_\delta=\int_0^{\infty}\frac{1}{1+u^{\frac{1}{\delta}}}{\rm d}u$. Interested readers are advised to refer to \cite{Haenggi2013} for more details. 
\vspace{-.35cm}
\section{Broadcast Approach} \label{sec:broadcast}
This is the main technical section of the paper, where we characterize the performance of the broadcast approach of \cite{Shamai1997,Shamai2003} in the ad-hoc network setup of the previous Section. We begin by introducing the broadcast approach briefly but encourage the readers to refer to \cite{Shamai1997,Shamai2003} for an in depth discussion in the context of a point-to-point link. Considering the most general setting for $S$, we treat it as a continuous random variable, which allows the transmitter to transmit, in parallel, infinitely many data layers parameterized by $S$, where $S$ can be interpreted as the continuous index of these layers. The power assigned to layer indexed by $S=s$ is $\rho(s){\rm d}s\geq 0$. 
Based on its channel state (given by $s$), the receiver will be able to decode a certain number of these layers. In particular, it will first decode the first layer while treating the signals transmitted over the other layers as interference. Next, in order to decode the second layer, it will first cancel out the signal transmitted over the first layer and then decode second layer treating the rest (third onwards) as interference. Likewise, the receiver will continue to successively decode the all layers indexed as $u\leq s$ by cancelling the signals in corresponding lower indexed layers. Thus, the decoding SINR for the layer indexed by $u$ becomes
\begin{align*}
\sinr_{bs}(u)=\frac{u\rho(u){\rm d}u}{1+uI(u)}
\end{align*}
where $I(u)=\int_u^\infty \rho(t){\rm d}t$.
The $I(s)$ is monotonically decreasing function of $s$ and the total transmit power allocated to all layers is $I(0)=\int_0^\infty\rho(s)ds=P$. 
Therefore, the {\em differential transmission rate} obtained by decoding layer indexed by $s$ is  
\begin{equation}
dR(s)=\ln\left(1+\frac{s\rho(s){\rm d}s}{1+sI(s)}\right)=\frac{s\rho(s){\rm d}s}{1+sI(s)}.
\end{equation}
The second equality is due to $\ln(1+f(x){\rm d}x)=f(x){\rm d}x$ for bounded $f(x)$.
The achievable rate for the channel state $s$ is an integration of the differential rates over all layers $u\leq s$:
\begin{equation}
R(s)=\int_0^s\frac{u\rho(u){\rm d}u}{1+uI(u)}.
\label{eq:RateCond_s}
\end{equation}
Note that the channel state $S$, defined in (2), is a function of the channel gain, co-channel interference and noise power.
Thus, the mean transmission rate for the typical receiver is
\begin{align}
R_{bs}(\lambda)&=\mathbb{E}_{\Phi}\left[\int_0^\infty R'(s|\Phi)(1-F_S(s|\Phi)){\rm d}s\right]\nonumber\\
&\stackrel{(a)}{=}\mathbb{E}_{\Phi}\left[\int_0^\infty \frac{s\rho(s)}{1+sI(s)}\prod_{\mathbf{x}_k\in\Phi}\frac{\exp(-sR_0^\alpha{\sigma}_N^2)}{1+sPR_0^\alpha\|\mathbf{x}_k\|^{-\alpha}}{\rm d}s\right]\nonumber\\
\stackrel{(b)}{=}\int_0^\infty& \frac{s\rho(s)}{1+sI(s)}\exp\left(-\pi\lambda R_0^2 (sP)^\delta \mathcal{Z}_\delta -sR_0^\alpha{\sigma}_N^2\right){\rm d}s
\label{eq:MeanRate_BA}
\end{align}
where step (a) follows using Lemma \ref{lemma:ConditionalCDF_S} and step (b) follows by exchanging the expectation with the integral and further using the Probability Generating Functional of a PPP~\cite{Haenggi2013}.

Next, we optimize $R_{bs} $ over the power distribution $\rho(s)$, equivalently over $I(s)$, under the total transmission power, i.e.,
\begin{equation}
R_{bs}(\lambda)= \max_{I(s)}\int_0^\infty \frac{s\rho(s)}{1+sI(s)}\exp\left(-G_\lambda s^\delta - s G_N\right){\rm d}s,
\label{eq:MeanRate_BA_max}
\end{equation}
where $G_N=R_0^\alpha{\sigma}_N^2$ and $G_\lambda=\pi\lambda R_0^2 P^\delta \mathcal{Z}_\delta$. Similar to \cite{Shamai2003}, using Euler-Lagrange equation, we obtain the function $I(s)$, for  which the functional in \eqref{eq:MeanRate_BA_max} is stationary, as follows
\begin{align}
I(s)=\begin{cases} 
\frac{1}{G_N s^2  + \delta G_\lambda s^{\delta+1}} - \frac{1}{s},~~~~&\text{if}~ s_0\leq s\leq s_1\\
0, &\text{otherwise},
\end{cases}
\label{eq:IS}
\end{align}
where $s_0$ is determined by $I(s_0)=P$  and  $s_1$ is determined by $I(s_1)=0$.
Further, using $\rho(s)=-\frac{{\rm d}}{{\rm d}s}I(s)$, we obtain the mean transmission rate as 
\begin{align}
R_{bs}(\lambda)=\int_{s_0}^{s_1} &\left(\frac{2G_Ns+\delta(\delta+1)G_\lambda s^\delta}{G_Ns+\delta G_\lambda s^\delta}-\left(G_Ns+\delta G_\lambda s^\delta\right)\right)\nonumber\\ 
&~~~~~~~~\times\frac{1}{s}\exp(-G_\lambda s^\delta -G_Ns){\rm d}s.\label{eq:MeanRate1}
\end{align}
Now, we study Eq. \eqref{eq:MeanRate1} in the following limiting cases.
\newline{\em Case 1) $\lambda\to 0$:} In the limiting case of $\lambda\to 0$, i.e. noise limited scenario, we obtain 
\begin{align*}
I(s)=\begin{cases}\frac{1}{G_N s^2}-\frac{1}{s}, & \text{if}~ s_0\leq s\leq s_1\\
0, &\text{otherwise}
\end{cases}
\end{align*}
where $s_0=\frac{\sqrt{G_N^2+4G_NP}-G_N}{2G_NP}$ and $s_1=\frac{1}{G_N}$. Further, by substituting $\rho(s)=-\frac{{\rm d}}{{\rm d}s}I(s)$ and $I(s)$ in \eqref{eq:MeanRate_BA} at $\lambda\to 0$, we get
\begin{equation}
R_{bs}=2(\text{Ei}(L_0)-\text{Ei}(L_1)) - (\exp(L_0)-\exp(L_1))
\label{eq:MeanRate_NoiseLimited}
\end{equation}
where $\text{Ei}(x)=\int_{x}^\infty \frac{\exp(-t)}{t}{\rm d}t$ is the exponential integral, $L_0=G_Ns_0$ and $L_1=G_Ns_1$. It can be observed that the mean transmission rate in \eqref{eq:MeanRate_NoiseLimited} is the same as the one obtained in \cite[Eq. (18)]{Shamai2003} with $R_0^2=1$ and $\sigma_N^2=1$. Refer to \cite[Eq. (20)]{Shamai2003} for the asymptotic behavior of Eq. \eqref{eq:MeanRate_NoiseLimited} w.r.t $P$.
\newline{\em Case 2) $\sigma_N^2\to 0$:} In the limiting case of $\sigma_N^2\to 0$, i.e. interference limited scenario, we obtain 
\begin{align*}
I(s)=\begin{cases} 
\frac{1}{\delta G_\lambda s^{\delta+1}} - \frac{1}{s},~~~~&\text{if}~ s_0\leq s\leq s_1\\
0, &\text{otherwise},
\end{cases}
\label{eq:IS}
\end{align*}
where $s_1=(\delta G_\lambda)^{-\frac{1}{\delta}}$ and $s_0$ is the solution to $s_0^\delta(s_0P+1)-\frac{1}{\delta G_\lambda}=0.$ Further, by substituting $\rho(s)=-\frac{{\rm d}}{{\rm d}s}I(s)$ and $I(s)$ in \eqref{eq:RateCond_s} and \eqref{eq:MeanRate_BA} at $\sigma_N^2\to 0$, we get transmission rate $R(s)$ as 
\begin{equation}
\hspace{-.07cm}R(s)\hspace{-.1cm}=\hspace{-.1cm}\begin{cases}
\hspace{-.05cm}(\delta+1)\ln(\frac{s}{s_0})\hspace{-.05cm}-G_\lambda[s^\delta-s_0^\delta],~\text{if}~ s\in[s_0,s_1]\\
\hspace{-.05cm}(\delta+1)\ln(\frac{s_1}{s_0})\hspace{-.05cm}-G_\lambda[s_1^\delta-s_0^\delta],~\text{if}~s_1<s\\
\hspace{-.05cm}0,~~~~~~~~~~~~~~~~~~~~~~~~~~~~~~~~~\text{otherwise},
\end{cases}
\label{eq:RateCond_BS}
\end{equation}
and the average transmission rate $R_{bs}(\lambda)$ as 
\begin{equation}
R_{bs}(\lambda)=\frac{\left(\text{Ei}(T_0)-\text{Ei}(T_1)\right)}{\left(1+{\delta}^{-1}\right)} - [\exp(-T_0)-\exp(-T_1)],
\label{eq:AverageRate_BS}
\end{equation} 
where ${T_0=G_\lambda s_0^\delta}$ and ${T_1=G_\lambda s_1^\delta}$. Since the noise-limited case is the same as the one studied in \cite{Shamai2003}, we will henceforth refer to the interference-limited case of $\sigma_N^2\to 0$. \\
{\em Remarks:} We have $T_1=G_\lambda s_1^\delta=\frac{1}{\delta}$ and $T_0$ is a solution of $PG_\lambda^{-\frac{1}{\delta}}T_0^{1+\frac{1}{\delta}}+T_0-\frac{1}{\delta}$. Therefore, using $G_\lambda=\pi\lambda R_0^2P^\delta\mathcal{Z}_\delta$, it is evident that the average transmission rate $R_{bs}(\lambda)$ is independent of the transmission power $P\in\mathbb{R}_+$.  Besides, the fact that $T_0\to T_1=\frac{1}{\delta}$ (or, $s_0\to 0$ and $s_1\to 0$) as $\lambda\to\infty$ implies that the $R_{bs}(\lambda)\to 0$ as $\lambda\to\infty$. This implies that there is a single layer in the limiting case of $\lambda$ such that $\rho(s)=P$ if $s=0$ and $\rho(s)=0$ otherwise. This means that the broadcast approach reduces to the outage strategy with $\sinr$ threshold $\beta=0$ in a highly dense network. 

Now coming to the outage strategy, from Eq. \eqref{eq:Rate_OS}, it can be observed that the average achievable rate under the standard outage strategy depends on threshold $\beta$. For $\sigma_N^2\to 0$, the maximum average rate can be expressed as
\begin{equation}
R_{os}(\lambda,\beta_{\text{opt}})=\exp\left(-\tilde{G}_\lambda \beta_{\text{opt}}^\delta\right)\ln\left(1+\beta_{\text{opt}}\right),
\label{eq:AchivableRate_OS}
\end{equation}
where $\tilde{G}_\lambda=\pi\lambda R_0^2\mathcal{Z}_\delta$ and $\beta_{\text{opt}}$ solves the equation
$$\beta_{\text{opt}}^{\delta-1}\left(1+\beta_{\text{opt}}\right)\ln\left(1+\beta_{\text{opt}}\right)=\left(\delta \tilde{G}_\lambda\right)^{-1}.$$
\vspace{-1.1cm}
\subsection{Variance of the Transmission Rate}
\vspace{-.1cm}
We now compute the variance of transmission rate under the two strategies. In Section~\ref{sec:ResultsandDiscussion}, we will demonstrate that the broadcast approach reduces variance in the transmission rate. Starting with the broadcast approach, we can obtain the second moment of $R(s)$ as follows
\begin{align*}
&R_{bs,2}(\lambda)=\int_{s_0}^{s_1}R(s)^2f_S(s){\rm d}s + R(s_1)^2\int_{s_1}^{\infty}f_S(s){\rm d}s\\
&=R(s)^2F_S(s)\big |_{s_0}^{s_1} - 2\int_{s_0}^{s_1}R(s)R'(s)F_S(s){\rm d}s \\
&~+ R(s_1)^2(1-F_S(s_1))\\
&\stackrel{(a)}{=}R(s_1)^2 - 2\int_{s_0}^{s_1}R(s)R'(s)F_S(s){\rm d}s\\
&\stackrel{(b)}{=}R(s_1)^2 +2\left(\Gamma(2,T_0)-\Gamma(2,T_1)\right) - G_\lambda^2\left(s_1^{2\delta}-s_0^{2\delta}\right)\\
&~~+\left(s_1^\delta-s_0^\delta\right)\left(2G_\lambda^2s_0^\delta+2(\delta+1)G_\lambda\ln\left({s_1}/{s_0}\right)\right) \\
&~~-\left(\exp\left(-T_0\right) - \exp\left(T_1\right)\right)\left(2\delta^{-1}(\delta+1)+2G_\lambda s_0^\delta\right)\\
&~~+2\delta^{-1}(\delta+1)\left(G_\lambda s_0^\delta-1\right)\left(\text{Ei}(T_0)-\text{Ei}(T_1)\right)\\
&~~+(\delta+1)\ln(s_1/s_0)\left(2\exp(-T_1)-(\delta+1)\ln(s_1/s_0)\right)\\
&~~+2(\delta+1)^2\int_{s_0}^{s_1}\frac{1}{s}\ln(s/s_0)\exp\left(-G_\lambda s^\delta\right){\rm d}s\numberthis
\label{eq:2ndMomentRate_BS}
\end{align*}
where step (a) follows from $R(s_0)=0$ and step (b) follows by substituting $R(s)$ from \eqref{eq:RateCond_BS}, and using $R'(s)=\frac{1}{s}\left((\delta+1) -\delta G_\lambda s^\delta\right)$, $F_S(s)=1-E_\Phi[\mathbb{P}[S>s|\Phi]]=1-\exp\left(-G_\lambda s^\delta\right)$ for $\sigma_N^2\to 0$, where $\mathbb{P}[S>s|\Phi]$ is given in Lemma \ref{lemma:ConditionalCDF_S}, followed by some mathematical manipulations.
Finally, using \eqref{eq:AverageRate_BS} and \eqref{eq:2ndMomentRate_BS}, the variance of the transmission rate under broadcast approach can be determined as follows
\begin{equation}
\sigma_{R_{bs}}^2(\lambda) = R_{bs,2}(\lambda) - R_{bs}(\lambda)^2.
\label{eq:VarianceRate_BS}
\end{equation}

The second moment of transmission rate under outage strategy can be obtained as 
$R_{os,2}(\lambda,\beta_{\text{opt}})=\exp(-\tilde{G}_\lambda\beta_\text{opt}^\delta)\ln^2(1+\beta_\text{opt})$.
Thus, using \eqref{eq:AchivableRate_OS}, the variance of transmission rate under outage strategy becomes
\begin{equation}
\sigma_{R_{os}}^2(\lambda,\beta_{\text{opt}})=R_{os,2}(\lambda,\beta_{\text{opt}}) -  R_{os}(\lambda,\beta_\text{opt})^2.
\end{equation}
\begin{figure}
\centering\vspace{-.4cm}
	\includegraphics[width=.78\columnwidth]{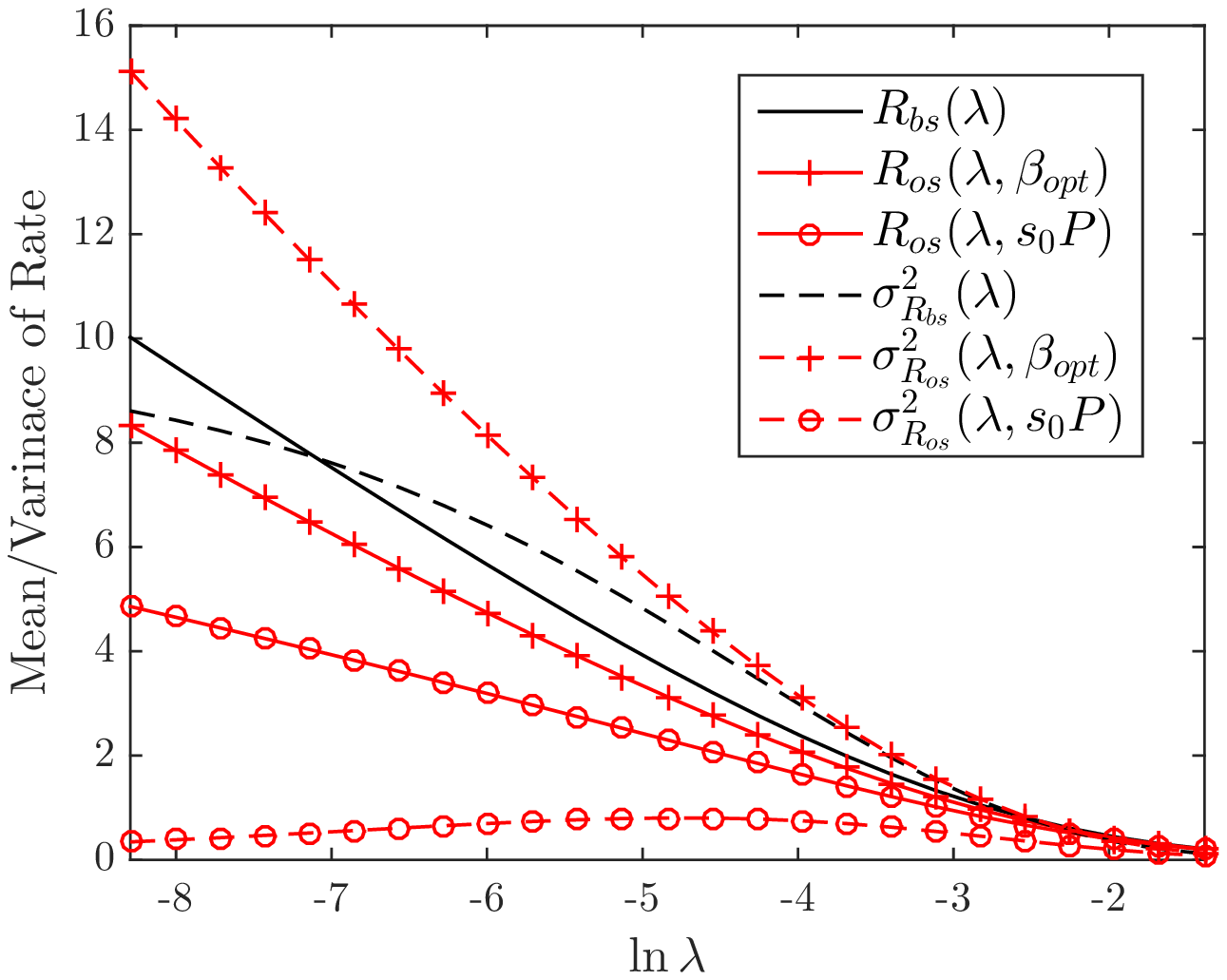}\\
    \includegraphics[width=.78\columnwidth]{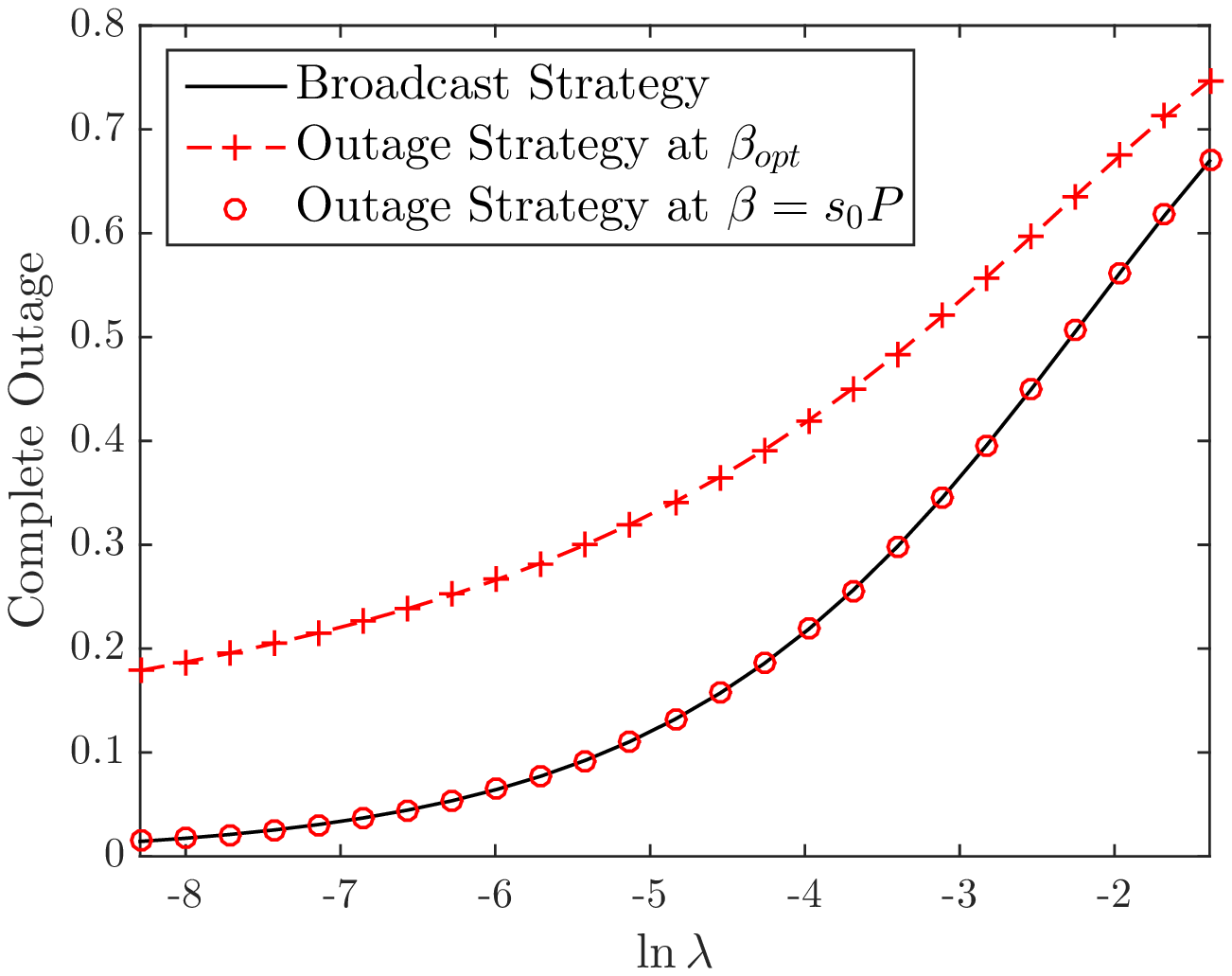} \vspace{-.4cm} 
  \caption{Top: mean and variance of achievable rate for $R_0=1$ and $\alpha=4$. Bottom: complete outage probability for $R_0=1$ and $\alpha=4$. }\vspace{-.5cm}
  \label{fig:MenaVar_Rate_and_OutageProbbaility}
\end{figure}
\vspace{-1cm}
\subsection{Transmission Capacity}
First define the rate outage probability as
$$q(\lambda)=\mathbb{P}[R(s)<\xi],$$
which is a continuous and monotonically increasing function of $\lambda$. On similar lines as~\cite{WeberAndrews2010}, we define TC as
\begin{equation}
c(\epsilon)=q^{-1}(\epsilon)R_{bs}(q^{-1}(\epsilon)),
\end{equation}
where $\epsilon\in (0,1]$. Note that $q^{-1}(\epsilon)$ is the maximum spatial density of transmitters associated with the rate outage $\epsilon$.  
In words, TC represent to the number of nats transmitted per unit time per unit area such that the transmission rate of a typical user is above threshold $\xi$ with probability $1-\epsilon$.
In the following theorem, we characterize the achievable transmission capacity under broadcast approach. 
\begin{theorem}
\label{thm:TransmissionCapacity_BS}
For $\xi<R(s_1)$ and $\mathcal{H}_\lambda<\exp(-1)$, the transmission capacity is $c(\epsilon)=\lambda_\epsilon R_{bs}(\lambda_\epsilon)$ where $R_{bs}(\lambda_\epsilon)$ is given by \eqref{eq:AverageRate_BS} and
\begin{equation}
\lambda_\epsilon=\left\{\lambda\in\mathbb{R}_+\mid \exp\left(\frac{\delta+1}{\delta}W\left(-\mathcal{H}_\lambda\right)\right)
<1-\epsilon\right\}
\label{eq:TransmissionCapacity_BS}
\end{equation}
 \begin{equation}
\text{where}~\mathcal{H}_\lambda=\frac{\delta }{\delta+1}G_\lambda s_0^\delta\exp\left(\frac{\delta}{\delta+1}\left(\xi-G_\lambda s_0^\delta\right)\right),~~~~~~~~~~~~~\nonumber
\end{equation}
$G_\lambda = \pi\lambda R_0^2 P^\delta \mathcal{Z}_\delta$ and $W(\cdot)$ is Lambert-W function.
\end{theorem}
\begin{proof}
From \eqref{eq:RateCond_BS}, we have $\mathbb{P}[R(s)\leq \xi]=1$ for $\xi>R(s_1)$. Now, for $\xi<R(s_1)$, we have
\begin{align*}
\mathbb{P}[R(s)\leq \xi]&=\mathbb{P}[R(s)\leq \xi|s\leq s_0]\mathbb{P}[s\leq s_0]\\
&+\mathbb{P}[R(s)\leq \xi|s> s_0]\mathbb{P}[s>s_0]\\
\stackrel{(a)}{=}\mathbb{P}[s\leq s_0]&+\mathbb{P}[R(s)\leq \xi|s> s_0]\mathbb{P}[s>s_0]\numberthis
\label{eq:Rate_Rth_Regions}
\end{align*}
where (a) follows due to $R(s)=0$ for $s<s_0$. Now, we determine $\mathbb{P}[R(s)\leq \xi|s> s_0]$ 
\begin{align*}
&=\mathbb{P}\left[(\delta+1)\ln\left({s}/{s_0}\right)-G_\lambda\left(s^\delta -s_0^\delta\right) \leq \xi\mid s> s_0\right]\\
&=\mathbb{P}\left[s\exp\left(-\frac{G_\lambda}{\delta+1}s^\delta\right)\leq s_0\exp\left(\frac{\xi-G_\lambda s_0^\delta}{\delta+1}\right)\mid s> s_0\right]\\
&\stackrel{(a)}{=}\mathbb{P}\left[-\frac{\delta G_\lambda}{\delta+1} s^\delta\exp\left(-\frac{\delta }{\delta+1}G_\lambda s^\delta\right)\geq  \right. \\
&~~~~~~~~~~\left. -\frac{\delta G_\lambda}{\delta+1}s_0^\delta\exp\left(\frac{\delta}{\delta+1}\left(\xi-G_\lambda s_0^\delta\right)\right)\mid s> s_0\right]\\
&\stackrel{(b)}{=}\mathbb{P}\left[s \leq \left(\frac{\delta+1}{\delta G_\lambda}W\left(-\mathcal{H}_\lambda\right)\right)^\frac{1}{\delta}\mid s> s_0\right]\numberthis
\label{eq:Rate_CDF_Rth_CondSo}
\end{align*}
where step (a) follows by raising both sides to power $\delta$ and multiplying by constant $-\frac{\delta G_\lambda}{\delta+1}$, and step (b) follows using the fact that $x=f^{-1}(x\exp(x))=W(x\exp(x))$, where $W(\cdot)$ is the Lambert-W function. The probability in \eqref{eq:Rate_CDF_Rth_CondSo} holds true for  $-\mathcal{H}_\lambda>-\exp(-1)$ as the root of $W(x)$ is complex valued for $x<-\exp(-1)$.
Finally, by substituting \eqref{eq:Rate_CDF_Rth_CondSo} along with $\mathbb{P}[S<s]=1-\exp(-G_\lambda s^\delta)$ into \eqref{eq:Rate_Rth_Regions} and solving further, gives $\lambda_\epsilon=q^{-1}(\epsilon)$ as in \eqref{eq:TransmissionCapacity_BS}, which completes the proof.
\end{proof}
\begin{figure}
\centering
	\vspace{-.2cm}
    \includegraphics[width=.78\columnwidth]{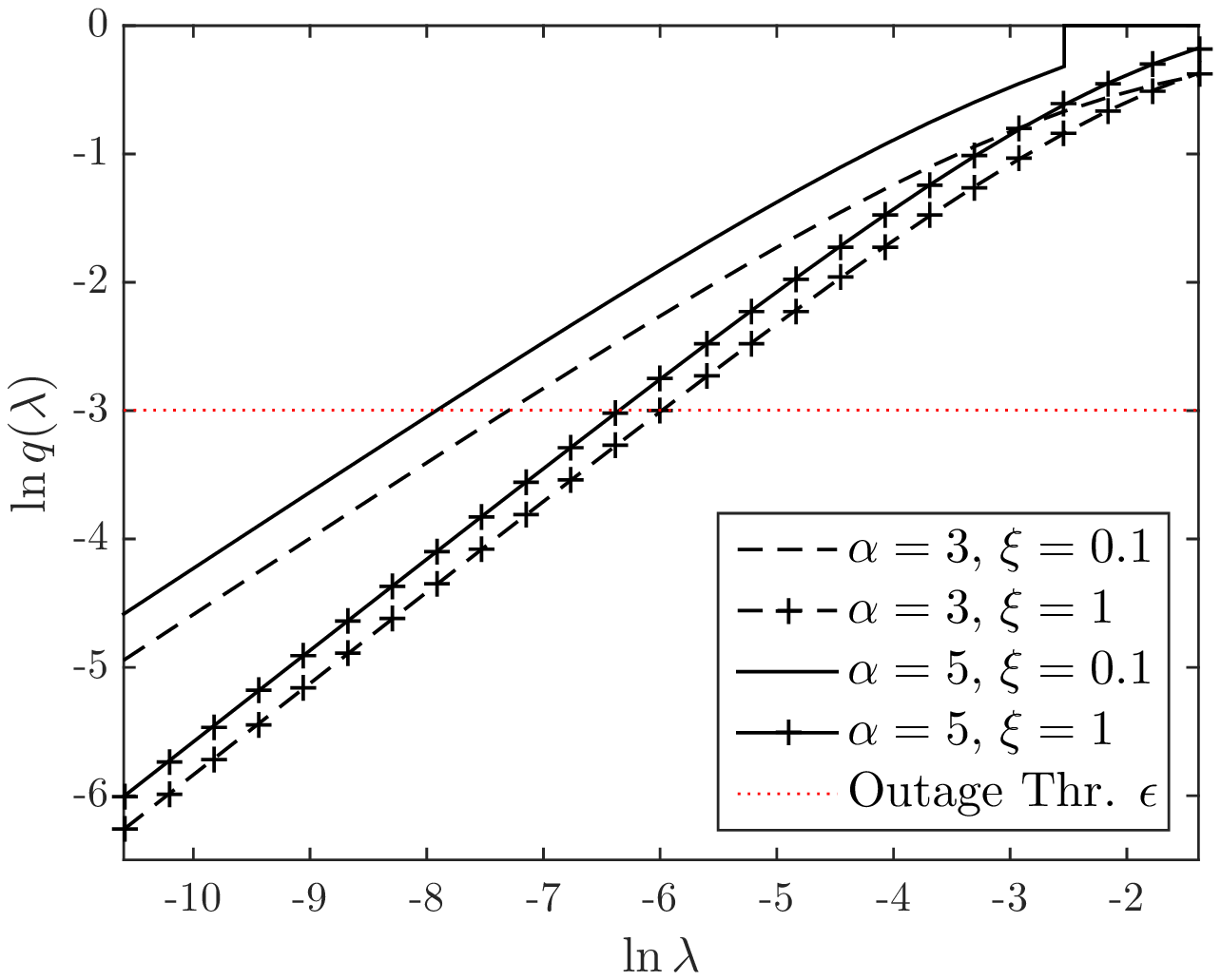}\\  
    \includegraphics[width=.78\columnwidth]{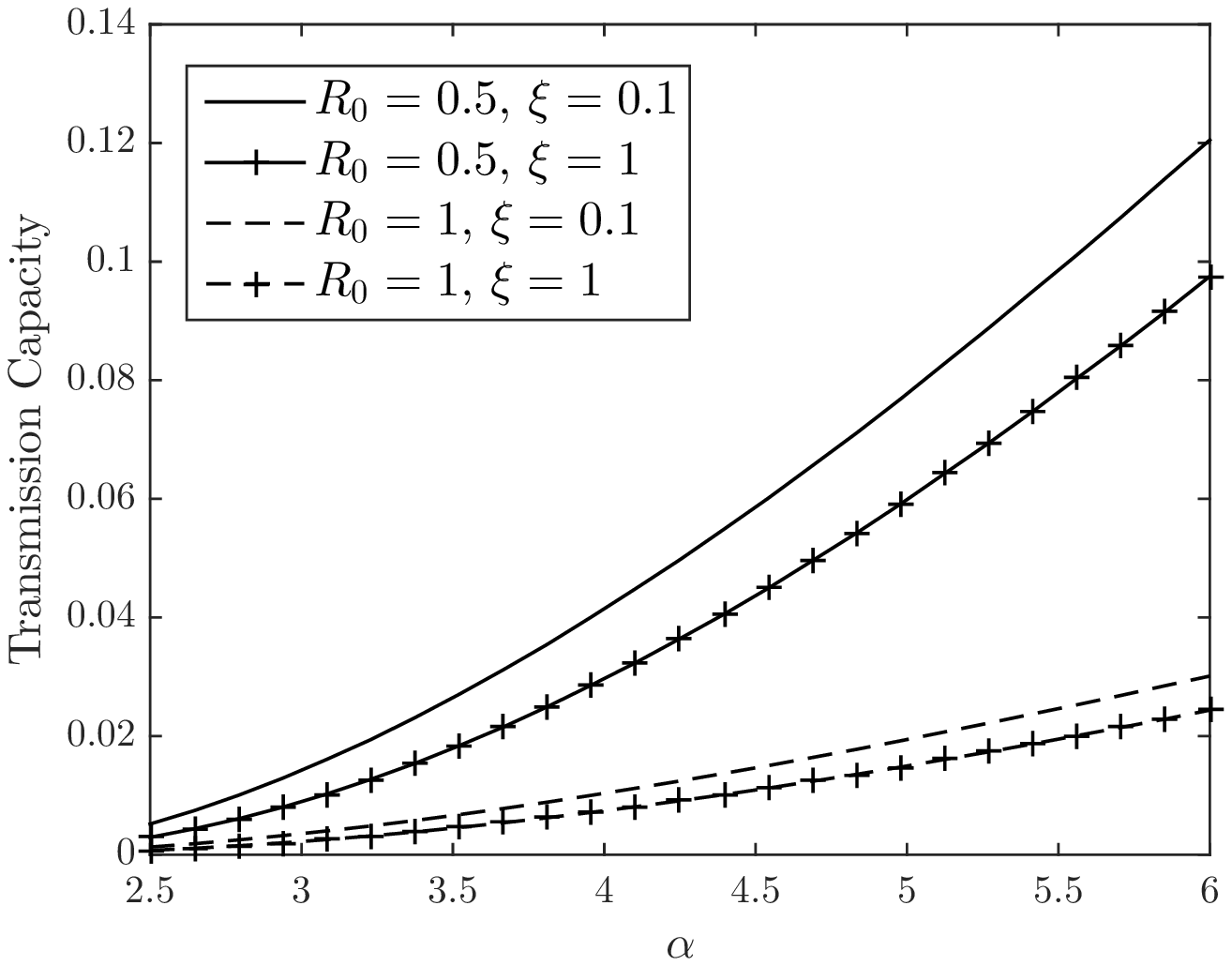} \vspace{-.3cm} 
  \caption{Top: rate outage vs. transmitter density $\lambda$ for $R_0=1$. Bottom: transmission capacity vs. pathloss exponent $\alpha$.}\vspace{-.4cm}
  \label{fig:TransmissionCapacity}
\end{figure}
\vspace{-.35cm}
\section{Results and Discussion}
\label{sec:ResultsandDiscussion}
In this section, we illustrate the network performance in terms of achievable transmission rate and transmission capacity for interference limited scenario i.e. $\sigma_N^2\to 0$.
For comparison purpose, we use the outage strategy with $\sinr$ thresholds $\beta=\beta_{\text{opt}}$, yielding maximum mean transmission rate, and $\beta=s_0P$, yielding the same \textit{complete outage} to the broadcast approach where we call the probability of a receiver not being able to decode a single layer as the \textit{complete outage}. 
 
From Fig. \ref{fig:MenaVar_Rate_and_OutageProbbaility}, it is clear that the broadcast approach provides better mean transmission rate compared to the outage strategy. Note that the gain in mean transmission rate further increases when the density of transmitters $\lambda$ is decreased. In addition, it can be seen that the broadcast approach yields less variation in transmission rate and better {\em complete outage} compared to the outage strategy with $\sinr$ threshold $\beta_{\text{opt}}$.  

Fig. \ref{fig:TransmissionCapacity} depicts the rate outage  probability and transmission capacity for outage threshold $\epsilon=0.05$ and $\xi\in\{0.1,1\}$. We observe that the $\mathcal{H}_\lambda<\exp(-1)$ for the dynamic rage of rate threshold $\xi$, which allows to evaluate $\mathbb{P}[R(s)<\xi]$ using the Lambert-W function. The jump in one of the curves in top figure is due to $\mathbb{P}[R(s)<\xi]=1$ for $\xi>R(s_1)$. As expected, the transmission capacity can be observed to be increasing with $\alpha$ and decreasing with $\xi$. 
\vspace{-.15cm}
\section{Conclusion}
In this letter, we have analyzed the achievable transmission rate in wireless networks under broadcast approach, with continuum of transmission layers, wherein the channel state information is assumed to be known at the receiver only. We have derived the mean and variance of the transmission rate under broadcast approach with optimal power distribution across the continuum of layers. Numerical results indicate that the broadcast approach provides improved achievable transmission rate at lower variation compared to optimally configured standard outage strategy. In addition, we also derived the rate distribution which allowed us to characterize the transmission capacity of the network under broadcast strategy. A useful extension of this work is the inclusion of traffic fluctuations to study {\em traffic-adaptive} broadcast approach.  

\end{document}

%% file: notation.tex
\def\nb0{{\mathbf{0}}}
\def\nb1{{\mathbf{1}}}

\newtheorem{lemma}{Lemma}

\newtheorem{theorem}{Theorem}

\def\sinr{\mathtt{SINR}}			

